\documentclass{article}

\usepackage{etex}
\usepackage{amssymb}
\usepackage{amsmath}
\usepackage{amsthm}
\usepackage{url}

\usepackage{mathtools}
\usepackage{stmaryrd}
\usepackage{color}
\usepackage{comment}

\usepackage[all]{xy}
\SelectTips{cm}{}

\theoremstyle{theorem}
\newtheorem{theorem}{Theorem}[section]

\newtheorem{lemma}[theorem]{Lemma}
\newtheorem{corollary}[theorem]{Corollary}
\theoremstyle{definition}
\newtheorem{definition}[theorem]{Definition}
\newtheorem{example}[theorem]{Example}
\newtheorem{remark}[theorem]{Remark}

\newcommand{\bbN}{\mathbb{N}}

\newcommand{\bbR}{\mathbb{R}}

\newcommand{\calI}{\mathcal{I}}

\newcommand{\cat}[1]{\mathcal{#1}}

\newcommand{\bigland}{\bigwedge}

\newcommand{\id}{\mathrm{id}}

\newcommand{\Term}{\mathcal{T}}
\newcommand{\Free}{\mathcal{F}}

\newcommand{\abs}[1]{{\lvert #1 \rvert}}

\newcommand{\blank}{{-}} 

\newcommand{\mono}{\rightarrowtail}
\newcommand{\epi}{\twoheadrightarrow}

\newcommand{\e}{\varepsilon}

\mathcode`\<="4268 
\mathcode`\>="5269 
\mathchardef\gt="313E 
\mathchardef\lt="313C 

\renewcommand{\Im}{\mathop{\mathrm{Im}}\nolimits}
\renewcommand{\subset}{\subseteq}

\title{Varieties of Metric and Quantitative Algebras}
\author{Wataru Hino \\ The University of Tokyo, Japan \\
        \url{wataru@is.s.u-tokyo.ac.jp}}
\date{\today}

\begin{document}
\maketitle

\begin{abstract}
  \emph{Metric algebras} are metric variants of $\Sigma$-algebras.
  They are first introduced in the field of universal algebra to deal with
  algebras equipped with metric structures such as normed vector spaces.
  Recently a similar notion of \emph{quantitative algebra} is
  used in computer science to formulate computational effects of probabilistic programs.
  In this paper we show that \emph{varieties of metric algebras}
  (classes defined by a set of \emph{metric equations}) are exactly
  classes closed under (metric versions of) subalgebras, products and quotients.
  This result implies the class of normed vector spaces cannot be defined
  by metric equations, differently from the classical case
  where the class of vector spaces is equational.
  This phenomenon suggests that metric equations are not a very natural class of formulas
  since they cannot express such a typical class of metric algebras.
  Therefore we need a broader class of formulas to acquire an appropriate
  metric counterpart of the classical variety theory.
\end{abstract}

\section{Introduction}

The theory of \emph{quantitative algebra} is a quantitative variant of
the theory of universal algebra,
where an atomic formula is of the form $s =_{\e} t$ using a non-negative real $\e \ge 0$.
This theory is introduced in \cite{Mardare2016}
to model computational effects of probabilistic computations
as a quantitative counterpart of the theory of algebraic effect \cite{PlotkinPower2003}.
In \cite{Mardare2016} they
deal with classes of metric algebras defined by \emph{basic quantitative inferences},
which are formulas of the form $\bigland_{i = 1}^{n} x_i =_{\e_i} y_i \to s =_{\e} t$,
where $x_i$ and $y_i$ are restricted to variables.

They introduce a complete deductive system for quantitative algebras
and show that various well-known metric constructions such as the Hausdorff metric,
the Kantrovich metric and the Wasserstein metric naturally arise as free quantitative algebras
with suitable implicational axioms.
The theory of quantitative algebra is also applied
to the axiomatization of the behavioral distance \cite{Bacci2016}.

The idea of using indexed binary relations to axiomatize metric structures
is already in the literature of universal algebra
\cite{Weaver1995, Khudyakov2003} under the name of \emph{metric algebra}.
It is a slightly wider notion than that of quantitative algebra
in the sense that operations in metric algebras are not required to be non-expansive.
The authors of \cite{Weaver1995, Khudyakov2003} prove
the continuous version of the characterization theorem
for \emph{quasivarieties}, i.e. classes of algebras defined by implications,
and the decomposition theorem corresponding to the one in the classical theory.
They suggest this framework can be used to reproduce the Hahn-Banach theorem,
a classical result on normed vector spaces.

Our work follows this line; we will study universal algebraic
and model-theoretic properties of quantitative and metric algebras.
Our main result is the metric version of Birkhoff's HSP theorem.
The proof is analogous to the classical case; we first construct
free algebras, and then take their quotients.

Another line of research comes from the field of continuous logic,
which is a variant of first-order logic that uses continuous truth values
rather than Boolean values.
The authors of \cite{Yaacov2007} use continuous logic
to formulate metric structures; they interpret the equality predicate as
distance functions.
This formulation is quite natural in the sense that it admits various model theoretic results
like compactness, types, quantifier elimination and stability.
For future work, we would like to study the relationship
between continuous logic and metric and quantitative algebras.

The authors of \cite{JoyOfCats} gives a categorical generalization of the HSP theorem
based on factorization systems, and it would also be interesting to check
whether our work is an instance of this formulation.

\section{Preliminaries}
\label{sec:preliminaries}

In this section, we will briefly review some basic notions
that we need before we introduce metric algebra.
First we consider some notions about metric spaces.

\begin{definition}
  \begin{itemize}
    \item A \emph{(extended) metric space} is a tuple $(X, d)$ of a set $X$ and a distance function
    $d \colon X \times X \to [0, \infty]$ that satisfies
    $d(x, y) = 0 \Leftrightarrow x = y$, $d(x, y) = d(y, x)$ and $d(x, y) + d(y, z) \ge d(x, z)$.

    \item A map $f \colon X \to Y$ between metric spaces is \emph{non-expansive}
    if it satisfies $d(f(x), f(y)) \le d(x, y)$ for each $x, y \in X$.

    \item For a family $(X_i, d_i)_{i \in I}$ of metric spaces,
    its \emph{product} is defined by $\left( \prod_{i} X_i, d \right)$ where
     $d((x_i)_i, (y_i)_i) = \sup_{i \in I} d_i(x_i, y_i)$
     and $d$ is called the \emph{sup metric}.
  \end{itemize}
\end{definition}

Note that we admit infinite distances (hence called \emph{extended})
because the category of extended metric spaces is categorically more amenable
than that of ordinary metric spaces; it has coproducts and arbitrary products.
Moreover a set can be regarded as a discrete metric space,
where every pair of two distinct points has an infinite distance.

\vspace{1ex}

In this paper, we denote $d(x, y) \le \e$ by $x =_{\e} y$.
To consider a metric structure as a family of binary relations
works well with various metric notions;
e.g. $f \colon X \to Y$ is non-expansive if and only if $x =_\e y$ implies
$f(x) =_\e f(y)$ for each $x, y \in X$ and $\e \ge 0$.
The sup metric of the product space $\prod_{i \in I} X_i$ is also compatible with
this relational view of metric spaces; it is characterized by
$(x_i)_i =_\e (y_i)_i \iff x_i =_{\e} y_i$ for all $i \in I$.

We adopt the sup metric rather than other metrics
(e.g. the 2-product distance) for the product of metric spaces.
One reason is the compatibility with the relational view above.
Another reason is that it corresponds to the product
in the category of extended metric spaces and non-expansive maps.

Recall that the sup metric does not always give rise to the product topology:
the product of uncountably many metrizable spaces is not in general metrizable.

\vspace{1ex}

The other ingredient of metric algebras is $\Sigma$-algebras in universal algebra.
Let $\Sigma$ be an algebraic signature, i.e. a set with an arity map
$\abs{\blank} \colon \Sigma \to \bbN$.

\begin{definition}
  \begin{itemize}
    \item A \emph{$\Sigma$-algebra} is a tuple $A = (A, (\sigma^A)_{\sigma \in \Sigma})$
    where $A$ is a set endowed with an operation $\sigma^A \colon A^\abs{\sigma} \to A$ for
    each $\sigma \in \Sigma$.

    \item A map $f \colon A \to B$ between $\Sigma$-algebras is a \emph{$\Sigma$-homomorphism}
    if it preserves all $\Sigma$-operations, i.e.
    $f(\sigma^A(a_1, \ldots, a_{\abs{\sigma}})) = \sigma^B(f(a_1), \ldots, f(a_{\abs{\sigma}}))$
    for each $\sigma \in \Sigma$ and $a_1, \ldots, a_{\abs{\sigma}} \in A$.

    \item A \emph{subalgebra} of a $\Sigma$-algebra $A$ is a subset of $A$ closed under $\Sigma$-operations,
    regarded as a $\Sigma$-algebra by restricting operations.
    A subalgebra is identified with (the isomorphic class of) a pair $(B, i)$ of a $\Sigma$-algebra
    and an injective homomorphism $i \colon B \to A$.

    \item The \emph{product} of $\Sigma$-algebras $(A_i)_{i \in I}$ is the direct product of
    the underlying sets endowed with the pointwise $\Sigma$-operations.

    \item A \emph{quotient} (or a \emph{homomorphic image}) of a $\Sigma$-algebra $A$ is
    a pair $(B, \pi)$ where $B$ is a $\Sigma$-algebra
    and $\pi \colon A \epi B$ is a surjective homomorphism.

    \item Given a set $X$, the \emph{set $\Term_{\Sigma} X$ of $\Sigma$-terms over $X$} is inductively defined as follows:
    each $x \in X$ is a $\Sigma$-term (called a \emph{variable}), and
    if $\sigma \in \Sigma$ and $t_1, \ldots, t_{\abs{\sigma}}$ are $\Sigma$-terms,
    then $\sigma(t_1, \ldots, t_{\abs{\sigma}})$ is a $\Sigma$-term.

    The set $\Term_{\Sigma} X$ is endowed with a natural $\Sigma$-algebra structure
    and this algebra is called the \emph{free $\Sigma$-algebra over $X$}.
    It satisfies the following universality:
    for each $\Sigma$-algebra $A$, a map $v \colon X \to A$ uniquely extends
    to a homomorphism $v^{\sharp} \colon \Term_{\Sigma} X \to A$.

    \item
    A \emph{$\Sigma$-equation} is a formula $X \vdash s = t$
    where $X$ is a set and $s$ and $t$ are $\Sigma$-terms.
    If $X$ is obvious from the context, we just write $s = t$.

    A $\Sigma$-algebra $A$ \emph{satisfies} an equation $X \vdash s = t$
    (denoted by $A \models s = t$)
    if and only if $v^{\sharp}(s) = v^{\sharp}(t)$ holds for any map $v \colon X \to A$.

    \item A class $\cat{K}$ of $\Sigma$-algebras is a \emph{variety} if there is a set of
    equations $\calI$ such that $\cat{K}$ is the class of $\Sigma$-algebras that satisfy all the equations in $\calI$.
  \end{itemize}
\end{definition}

If the signature $\Sigma$ is obvious from the context,
we omit the prefix $\Sigma$ and just say homomorphism, equation, etc.

\vspace{1ex}

The following theorem is a fundamental theorem in universal algebra proved by Birkhoff.
It states that the property of being a variety is equivalent to a certain closure property.
Our goal is to prove the metric version of this theorem.

\begin{theorem}[Birkhoff's HSP theorem]
  A class $\cat{K}$ of $\Sigma$-algebras is a variety
  if and only if $\cat{K}$ is closed under subalgebras, products and quotients. \qed
\end{theorem}

\section{Variety Theorem For Metric Algebras}
\label{sec:variety-malg}

By combining metric structures and $\Sigma$-algebraic structures,
we acquire the definitions of metric algebra and quantitative algebra.
They go as follows.

\begin{definition}[Weaver \cite{Weaver1995}] \leavevmode
  \begin{itemize}
    \item A \emph{metric $\Sigma$-algebra} is a tuple $A = (A, d, (\sigma^A)_{\sigma \in \Sigma})$
    where $(A, d)$ is a metric space and $(A, (\sigma^A)_{\sigma \in \Sigma})$
    is a $\Sigma$-algebra.

    It is called a \emph{quantitative algebra} if each
    $\sigma^{A} \colon A^\abs{\sigma} \to A$ is non-expansive
    (where $A^\abs{\sigma}$ is equipped with the sup metric).

    \item A map $f \colon A \to B$ between metric algebras is
    an \emph{M-homomorphism} if $f$ is $\Sigma$-homomorphic and non-expansive.

    \item An \emph{M-subalgebra} $B$ of a metric algebra $A$
    is a subalgebra of $A$ equipped with the induced metric.
    It can be identified with an embedding, i.e.
    a pair $(B, i)$ of a metric algebra $B$ and
    an isometric homomorphism $i \colon B \mono A$.

    \item The \emph{M-product} of metric algebras is
    the product as $\Sigma$-algebras equipped with the sup metric.

    \item An \emph{M-quotient} of a metric algebra $A$ is a pair $(B, \pi)$
    where $B$ is a metric algebra and $\pi \colon A \epi B$ is a surjective $M$-homomorphism.
    It is a \emph{Q-quotient} if $A$ and $B$ are quantitative algebras.

    \item
    For a set $X$ of variables,
    an \emph{M-equation} (or an \emph{atomic inequality} \cite{Weaver1995})
    is a formula of the form $X \vdash s =_{\e} t$
    where $X$ is a set, $s, t \in \Term_{\Sigma} X$ and $\e \ge 0$.
    We omit $X \vdash$ if $X$ is obvious.

    A metric algebra $A$ \emph{satisfies} an M-equation
    $X \vdash s =_{\e} t$ (denoted by $A \models s =_{\e} t$)
    if it satisfies $d^A(v^{\sharp}(s), v^{\sharp}(t)) \le \e$ for each map $v \colon X \to A$.

    \item A class of metric (resp. quantitative) algebras is called a \emph{variety}
    of metric (quantitative) algebras if it is defined by a set of $M$-equations.
  \end{itemize}
\end{definition}

We do not define the notion of Q-subalgebra and Q-product because
M-subalgebras and M-products of quantitative algebras are quantitative algebras.

\vspace{1ex}

The choice of variable set $X$ is irrelevant in the definition
of the satisfaction of metric equations.
This follows from the following straightforward lemma.

\begin{lemma}\label{lem:renaming}
  Let $Y \subset X$ be sets, $Y \vdash s =_\e t$ be an M-equation,
  and $A$ be a metric algebra. Then the following conditions are equivalent:
  \begin{enumerate}
    \item $A$ satisfies $Y \vdash s =_{\e} t$.
    \item $A$ satisfies $X \vdash s =_{\e} t$. \qed
  \end{enumerate}
\end{lemma}

In the classical case, the following lemma on quotients of algebras is
a folklore theorem, e.g. found in Exercise \textsection 6 of \cite{burris1981course}.

\begin{lemma}\label{lem:quotient-univ}
  Let $A$ be a $\Sigma$-algebra and
  $p \colon A \to B$, $q \colon A \to C$ be its quotients.
  If $p(a) = p(b)$ implies $q(a) = q(b)$ for all $a, b \in A$,
  there is a unique homomorphism $h \colon B \to C$ such that $h \circ p = q$.
  Moreover $h$ is surjective.
\end{lemma}
\begin{proof}
  Such $h$ is unique by the surjectivity of $p$.
  We construct $h \colon B \to C$ as follows:
  given $u \in B$, take $a \in A$ such that $p(a) = u$ and define $h(u) = q(a)$.
  Then $h$ is well-defined: if we take another $b \in A$ with $p(b) = u$,
  by the fact $p(a) = p(b)$ and the assumption, we have $q(a) = q(b)$.
\end{proof}

The metric version of Lemma~\ref{lem:quotient-univ} also holds.

\begin{lemma}\label{lem:quotient-qalg}
  Let $A$ be a metric algebra and $p \colon A \to B$, $q \colon A \to C$ be its M-quotients.
  If $p(a) =_{\e} p(b)$ implies $q(a) =_{\e} q(b)$ for all $a, b \in A$ and $\e \ge 0$,
  there is a unique M-homomorphism $h \colon B \to C$ such that $h \circ p = q$,
  which is surjective.
\end{lemma}
\begin{proof}
  The uniqueness directly follows from the surjectivity of $p$, so we only have to
  show the existence.
  By setting $\e = 0$ we have $p(a) = p(b)$ implies $q(a) = q(b)$,
  thus by Lemma~\ref{lem:quotient-univ}
  we have a surjective homomorphism $h \colon B \to C$ satisfying $h \circ p = q$.
  We show that $h$ is non-expansive. Suppose $u, v \in B$ satisfy $u =_{\e} v$.
  Since $p$ is surjective, we can take $a, b \in A$ where $p(a) = u$ and $p(b) = v$.
  Using the assumption and $p(a) =_{\e} p(b)$, we have $q(a) =_{\e} q(b)$.
  Therefore $h(u) = h(p(a)) = q(a) =_{\e} q(b) = h(p(b)) = h(v)$ holds.
\end{proof}

We are now going to prove the variety theorem for metric and quantitative algebras.
As is the case of classical universal algebra, the proof relies on the existence
of free algebras over sets in a given class.
Note that, given a metric algebra $A$, the class of $M$-quotients of $A$ is small
up to isomorphisms.

\begin{theorem}\label{thm:free-alg}
  Let $\cat{K}$ be a class of metric algebras closed under
  M-products and M-subalgebras.
  Then for any set $X$, there is a metric algebra $\Free_{\cat{K}} X$ in $\cat{K}$
  (called the \emph{$\cat{K}$-free algebra over $X$})
  and a surjective homomorphism $F \colon \Term_{\Sigma} X \to \Free_{\cat{K}} X$ satisfying
  the following universality: for any metric algebra $A \in \cat{K}$ and
  a homomorphism $f \colon \Term_{\Sigma} X \to A$,
  there exists a unique $M$-homomorphism $h \colon \Free_{\cat{K}} X \to A$
  such that $h \circ F = f$.
  Moreover for $s, t \in \Term_{\Sigma} X$, we have $F(s) =_\e F(t)$ if and only if
  $A \models s =_{\e} t$ for all $A \in \cat{K}$.
\end{theorem}
\begin{proof}
  Fix a set $X$.
  Let $\cat{G} = \{g_i \colon \Term_{\Sigma}{X} \epi B_i \in \cat{K} \}_{i \in I}$
  be the set of (isomorphic classes of) quotients of $\Term_{\Sigma}{X}$ to metric algebras in $\cat{K}$.
  We will denote its product by $G = (g_i)_i \colon \Term_{\Sigma}{X} \to \prod_i B_i$
  and the image factorization of $G$
  by $i \circ F \colon \Term_{\Sigma}{X} \epi \Im(G) \mono \prod_i B_i$.
  We will show $\Im(G)$ is the $\cal{K}$-free algebra over $X$.

  Assume $A$ is a metric algebra in $\cat{K}$ and $f \colon \Term_{\Sigma} X \to A$ is a homomorphism.
  Let $e \circ p \colon \Term_{\Sigma}{X} \epi \Im(f) \mono A$ be the image factorization of $f$.
  Since the isomorphic class of $p$ belongs to $\cal{G}$,
  there exists a projection map $\pi \colon \prod_i B_i \to \Im(f)$
  that is M-homomorphic and satisfies $\pi \circ G = p$.
  Thus we have a commuting M-homomorphism $h = e \circ \pi \circ i \colon \Im(G) \to B$.
  The uniqueness of $h$ follows from the surjectivity of $F$.
  Therefore $\Im(G)$ is the $\cal{K}$-free algebra over $X$.

  The last assertion follows from the universality:
  suppose $F(s) =_{\e} F(t)$, then for any $A \in \cat{K}$ and a map $v \colon X \to A$,
  there exists a $M$-homomorphism $h \colon \Free_{\cat{K}} X \to A$ such that $h \circ F = v^{\sharp}$.
  Since $F(s) =_{\e} F(t)$ and $h$ is non-expansive,
  we have $v^{\sharp}(s) = h(F(s)) =_{\e} h(F(t)) = v^{\sharp}(t)$,
  hence $A \models s =_{\e} t$. Conversely if $A \models s =_{\e} t$ for all $A \in \cat{K}$,
  by taking $A = \Free_{\cat{K}} X$ we conclude $F(s) =_\e F(t)$.
\end{proof}

\begin{remark}
  In \cite{Mardare2016} they give a syntactic construction of free algebras
  when $\cat{K}$ is a class defined by a family of \emph{basic quantitative inferences},
  i.e. formulas of the form $\bigland_{i=1}^{n} x_i =_{\e_i} y_i \to s =_{\e} t$.
  Such classes are closed under $M$-subalgebras and $M$-products,
  hence Theorem~\ref{thm:free-alg} can be applied.
  This gives a semantic construction of free algebras for those cases.
\end{remark}

Since the class of quantitative algebras is closed under M-subalgebras and M-products,
the quantitative version of the theorem directly follows.

\begin{corollary}
  A class $\cal{K}$ of quantitative algebras closed under
  M-products and M-subalgebras has the $\cal{K}$-free algebra over $X$ for any set $X$.
\end{corollary}

Now we can prove the metric variety theorem using free algebras.

\begin{theorem}\label{thm:hsp-theorem}
  A class $\cat{K}$ of metric algebras is a variety
  if and only if $\cat{K}$ is closed under M-products, M-subalgebras and M-quotients.
\end{theorem}
\begin{proof}
  The sufficiency is trivial; we will show the necessity.

  Let $\calI$ be the set of M-equations over finite variables
  that hold for all $A \in \cat{K}$.
  Suppose that $A \models \calI$. Let $X = A$ and $p \colon \Term_{\Sigma} X \to A$
  be the homomorphic extension of the identity map $\id_A \colon A \to A$.
  Note that $p$ is surjective.

  Let $F \colon \Term_{\Sigma} X \to \Free_{\cat{K}}(X)$ be the canonical map
  defined in Theorem~\ref{thm:free-alg}.
  By Lemma~\ref{lem:quotient-qalg},
  we only have to show $F(s) =_{\e} F(t)$ implies $p(s) =_{\e} p(t)$.

  Assume $F(s) =_{\e} F(t)$.
  By Theorem~\ref{thm:free-alg}, any metric algebra in $\cat{K}$
  satisfies $X \vdash s =_{\e} t$.
  Since there occurs finitely many variables in the terms $s$ and $t$,
  we can view $s =_{\e} t$ as an M-equation over some finite set $Y$.
  By Lemma~\ref{lem:renaming}, any metric algebra in $\cat{K}$
  also satisfies $Y \vdash s =_{\e} t$.
  Therefore $Y \vdash s =_\e t \in \calI$, and then $A \models s =_\e t$.
  In particular we can conclude $p(s) =_\e p(t)$.
\end{proof}

In the same manner, we can also show the quantitative variety theorem.

\begin{corollary}
  A class $\cat{K}$ of quantitative algebras is defined by metric equations
  if and only if $\cat{K}$ is closed under M-products, M-subalgebras and Q-quotients.
\end{corollary}

Applying Theorem~\ref{thm:hsp-theorem},
we can prove that the class of normed vector spaces is not a variety of metric algebras
(for the signature of vector spaces).

\begin{example}
  For the signature $\Sigma = \{{+}, 0, (\alpha {\cdot}) \}_{\alpha \in \bbR}$,
  the class $\cat{N}$ of normed vector spaces is a \emph{(continuous) quasivariety}
  of metric $\Sigma$-algebras \cite{Weaver1995, Khudyakov2003},
  but it is not a variety. Indeed consider $\bbR \in \cat{N}$ and
  let $R'$ be a metric algebra that has the same algebraic structure as $\bbR$
  but whose distance is defined by $d(x, y) = \abs{\tanh(y) - \tanh(x)}$.
  The identity map $f \colon \bbR \to R'$ is an M-quotient
  while $R' \not\in \cat{N}$. Therefore $\cat{N}$ is not closed under M-quotient,
  hence not a variety.
\end{example}

The class of normed vector space is a prototypical example of classes of metric algebras,
while the above discussion showed it cannot be expressed by M-equations.
To deal with such classes, we need to use more expressive formulas such as basic quantitative inferences.
We leave to future works to find the closure property that characterizes classes
defined by basic quantitative inferences.

\bibliographystyle{plain}
\bibliography{ref}

\end{document}